\documentclass[11pt]{article}
\usepackage{amsmath,amssymb,amsthm,mathtools}
\usepackage[margin=1in]{geometry}
\usepackage{microtype}
\usepackage{enumitem}
\usepackage[square, numbers, sort&compress]{natbib}

\usepackage{hyperref}
\usepackage{orcidlink}

\setlist{nosep}
\newtheorem{theorem}{Theorem}
\newtheorem{proposition}[theorem]{Proposition}
\newtheorem{lemma}[theorem]{Lemma}
\newtheorem{corollary}[theorem]{Corollary}
\theoremstyle{remark}
\newtheorem{remark}[theorem]{Remark}

\title{Efficient Defection: Overage-Proportional Rationing Attains the Cooperative Frontier}
\author{Florian Lengyel\,\orcidlink{0000-0003-1210-2058} \\
        \href{mailto:florian.lengyel@cuny.edu}{florian.lengyel@cuny.edu}}
\date{September 9, 2025}

\begin{document}
\maketitle

\begin{abstract}
We study a noncooperative $n$-player game of \emph{slack allocation}  in which each player $j$ has entitlement $L_j>0$ and chooses a claim $C_j\ge0$. Let $v_j=(C_j-L_j)_+$ (overage) and $s_j=(L_j-C_j)_+$ (slack); set $X=\sum_j v_j$ and $I=\sum_j s_j$. At the end of the period an overage-proportional clearing rule allocates cooperative surplus $I$ to defectors in proportion to $v_j$; cooperators receive $C_j$. We show: (i) the selfish outcome reproduces the cooperative payoff vector $(L_1,\dots,L_n)$; (ii) with bounded actions, defection is a weakly dominant strategy; (iii) within the $\alpha$-power family, the linear rule ($\alpha=1$) is the unique boundary-continuous member; and (iv) the dominant-strategy outcome is Strong Nash under transferable utility and hence coalition-proof \citep{BernheimPelegWhinston1987}. We give a policy interpretation for carbon rationing with a penalty collar.
\end{abstract}

\section{Introduction}

We study an $n$-player noncooperative ``slack allocation'' game. Each agent $j$ holds an entitlement $L_j>0$ and chooses a claim $C_j\ge 0$. Let the overage and slack be $v_j=(C_j-L_j)_+$ and $s_j=(L_j-C_j)_+$, with aggregates $X=\sum_j v_j$ and $I=\sum_j s_j$.%
\footnote{$X:=\sum_j v_j$ denotes total overage; when comparing to classic bankruptcy rules we write $C_{\mathrm{tot}}:=\sum_j C_j$ for total claims.}
At period end, a clearing rule allocates the cooperative surplus $I$ to defectors proportionally to their overage; cooperators receive their claims. The rule is budget balanced when scarcity binds ($X\ge I$) and treats cooperators as ``no-sucker-loss'': if $C_j\le L_j$ then $\pi_j=C_j$ regardless of others.

Our main result is that this proportional slack clearing implements the cooperative frontier in dominant strategies (under bounded actions): each player’s payoff equals their entitlement in equilibrium, even though the behavior is self-regarding defection. We show the dominant-strategy profile is robust to coalition deviations under transferable utility (coalition-proof in the sense of \citet{BernheimPelegWhinston1987}). We also characterize proportionality within a natural $\alpha$-power family: continuity at the $X=I$ boundary uniquely selects the linear rule $\alpha=1$ (Theorem~\ref{thm:uniq}).

We assume credible end-of-period enforcement of the clearing rule and observable claims/emissions. Dominance requires bounded actions $C_j\in[0,M]$; without bounds, best replies may exist only in the limit (Appendix~C).  Coalition-proofness is stated at the dominant-strategy profile under transferable utility.

\paragraph{Contributions.}
(i) \emph{Implementation by efficient defection.} With bounded actions, the max-claim action is a weakly dominant strategy; the induced outcome reproduces the cooperative payoff vector $(L_j)_j$ and is budget balanced when $X\ge I$. (ii) \emph{Robustness to collusion.} At the dominant-strategy profile, no coalition can Pareto-improve under TU; the profile is Strong Nash and hence Coalition-Proof \citep{BernheimPelegWhinston1987}. (iii) \emph{Characterization.} We consider a generalized $\alpha$-power family where the surplus $I$ is allocated proportionally to the $\alpha$-power of individual overages (i.e., $(v_j)^\alpha$). We show that the linear rule ($\alpha=1$) is the unique member of this family that is continuous at the boundary $X=I$. (iv) \emph{Policy reading.} As an end-of-period clearing mechanism with a penalty band, the design is compatible with forward trading and eliminates “wait-and-emit” arbitrage (Appendix D).

\paragraph{Relation to existing work.}
The paper intersects three literatures. First, in the \emph{claims/rationing} tradition (bankruptcy and uniform rationing), proportional rules are classically justified by axioms such as anonymity, consistency, and resource monotonicity \citep[see, e.g.,][]{Thomson2015,Thomson2003,ONeill1982,AumannMaschler1985,Moulin2000}. Our mechanism is noncooperative, budget balanced under scarcity, treats cooperators lexicographically (no-sucker-loss), and yields a new characterization via boundary continuity. Second, in \emph{congestion/CPR} and network allocation, proportional sharing appears via prices and progressive filling \citep[e.g.,][]{Kelly1997,LowLapsley1999}, but agent payoffs there are typically price-mediated and not dominance-implementable. Our rule is price-free, direct, and dominance-implementable under bounds. Third, on \emph{coalition-proofness}, we work within the \citet{BernheimPelegWhinston1987} framework and show the DS outcome is Strong Nash under TU because coalition surplus “leaks” to nonmembers via proportional coverage.

\subsection*{Related literature}

\emph{Claims, bankruptcy, and rationing.} Classical bankruptcy/claims problems allocate a fixed estate to claimants under axioms such as anonymity, consistency, and resource monotonicity; proportional and related rules are characterized in this tradition \citep{ONeill1982,AumannMaschler1985,Moulin2000,Thomson2003,Thomson2015}. Our setting differs: actions are strategic, cooperators are guaranteed their claims (no-sucker-loss), and budget balance holds only when scarcity binds; within this design, boundary continuity selects proportionality.

\emph{Congestion/CPR and networks.} Proportional sharing appears in congestion control and progressive-filling allocations  \citep[e.g.,][]{Kelly1997,LowLapsley1999}; those models rely on prices and potential-game structures. We instead give a direct, price-free mechanism with dominance under bounds and coalition-proofness at equilibrium.

\emph{Coalition-proofness.} We adopt the coalition-proof Nash framework of \citet{BernheimPelegWhinston1987} and show the dominant-strategy outcome is Strong Nash under TU, hence coalition-proof, because coalition-generated surplus is diluted proportionally to overage, limiting the coalition’s net gain.

\paragraph{Roadmap.}
Section~\ref{sec:mechanism} defines the rule and states the budget identity. Section~\ref{sec:properties} gives the main properties (dominance under bounds, coalition-proofness, boundary characterization). Appendix~A develops the $\alpha$-family and the continuity uniqueness; Appendix~B proves coalition-proofness; Appendix~C provides the bounded-action regularization; Appendix~D gives the policy economics of the penalty band.

\section{Mechanism (slack allocation)}\label{sec:mechanism}

For each player $j$, define
\[
v_j:=(C_j-L_j)_+,\qquad s_j:=(L_j-C_j)_+,\qquad X:=\sum_{m=1}^n v_m,\qquad I:=\sum_{m=1}^n s_m.
\]
Here $(x)_+ := \max\{x,0\}$.

Define the cooperator and defector sets by $S:=\{j:\, C_j\le L_j\}$ and $D:=\{j:\, C_j> L_j\}$.

\emph{Assumption (costless claims).} Settlement payoffs are $\pi_j$; submitting a claim $C_j$ carries no magnitude-dependent cost. (If utility were $U_j=\pi_j-\epsilon C_j$ with $\epsilon>0$, maximal claiming $C_j=M$ would not be weakly dominant.)

\paragraph{Bounded actions.}
When we study dominant-strategy behavior we restrict claims to $C_j\in[0,M]$ with a common bound $M> \max_j L_j$;
this is the only place where boundedness is used (see \emph{Appendix C}).

For $\alpha=1$ (linear rule), defectors receive
\begin{equation}\label{eq:clear}
\widehat v_j=\begin{cases}
v_j,& X\le I,\\[3pt]
\dfrac{I}{X}\,v_j,& X>I,
\end{cases}
\qquad\text{and}\qquad
\pi_j=\begin{cases}
C_j,& C_j\le L_j,\\[2pt]
L_j+\widehat v_j,& C_j>L_j.
\end{cases}
\end{equation}
Aggregate payoffs satisfy
\begin{equation}\label{eq:budget}
\sum_j \pi_j \;=\; \sum_j L_j \;-\; \max\{I-X,0\},
\end{equation}
If $X<I$, the gap $I-X$ is unused surplus; if $X\ge I$ (scarcity binds), the rule is budget balanced.
At the cooperative profile $C=L$ we have $(X,I)=(0,0)$ and $\pi_j=L_j$;
at the all-defect profile ($I=0$) we again have $\pi_j=L_j$.

\paragraph{Design trade-off (incentives vs ex-post efficiency).}
The mechanism attains the cooperative frontier in equilibrium by tolerating off-equilibrium inefficiency: when $X<I$, the gap $I-X$ is discarded rather than rebated. This potential waste creates strong ex-ante incentives to claim aggressively; in the dominant-strategy outcome all agents claim $M$, yet realized payoffs equal $L_j$ and total welfare $\sum_j L_j$ is achieved. The result is an incentives-for-efficiency trade-off, not a free lunch.

\paragraph{Normative rationale for continuity at $X=I$.}
Boundary continuity eliminates settlement cliffs under measurement/reporting noise.
With small symmetric noise near $X=I$, any $\alpha\neq 1$ creates a boundary jump that yields a finite expected coverage bias even for small noise,
whereas the linear rule ($\alpha=1$) removes the jump so the bias vanishes with the noise and incentives are locally robust.

\paragraph{Axiomatic contrast to CPR/congestion.}
(i) \emph{No-Sucker-Loss}: if $C_j\!\le\!L_j$ then $\pi_j=C_j$ regardless of others. (ii) \emph{Scarcity-Budget-Balance}: if $X\!\ge\!I$ then $\sum_j\pi_j=\sum_j L_j$. Standard CPR/congestion models typically violate (i). The slack-allocation rule is the unique linear proportional member that satisfies both while remaining boundary-continuous.
This ``No-Sucker-Loss'' guarantee isolates cooperative agents from externalities created by over-claimants, a fairness property uncommon in standard CPR models.

\begin{remark}[Why the No-Sucker-Loss guarantee is atypical]\label{rem:NLS-atypical}
For comparison with classic rules, write $C_{\mathrm{tot}}:=\sum_j C_j$ for total claims (distinct from our $X=\sum_j v_j$, total \emph{overage}). In those models, even “cooperative” agents (those with $C_j\le L_j$) can see their payoffs reduced when the system is under stress.

(i) \textbf{Proportional rule on claims} \citep{ONeill1982,Thomson2015}.
When $C_{\mathrm{tot}}>I$, each agent receives a fraction $\lambda=I/C_{\mathrm{tot}}<1$ of their claim. Thus a cooperative agent with $C_j\le L_j$ receives $\lambda C_j<C_j$, violating NLS.

(ii) \textbf{Constrained equal awards (CEA)} \citep{AumannMaschler1985,Thomson2015}.
Awards are $a_j=\min\{C_j,\lambda\}$ with $\lambda$ chosen to exhaust the estate. With $C=(1,100,100)$ and $I=2$, $\lambda=2/3$ and the cooperative agent gets $a_1=2/3<1$, violating NLS.

(iii) \textbf{Network proportional fairness} \citep{Kelly1997,LowLapsley1999}.
Allocations are jointly determined by a coupled optimization; holding $C_j$ fixed, increasing other users’ demands can strictly decrease agent $j$’s allocation, so there is no analogue of NLS.

In contrast, \emph{overage-proportional rationing} applies reductions only to overages $(C_j-L_j)_+$. Any agent with $C_j\le L_j$ receives exactly $C_j$ regardless of others’ claims. These canonical families all fail NLS, indicating that the property is non-generic among standard rationing and congestion models.
\end{remark}

\section{Main properties}\label{sec:properties}
\begin{proposition}[Cooperative frontier reproduced]\label{prop:frontier}
If all defect, then $I=0$, $\widehat v_j=0$, and $\pi_j=L_j$ for all $j$.
If all cooperate, then $X=0$, $\pi_j=C_j$, and at $C=L$ the payoff vector is $(L_j)_j$.
\end{proposition}

\begin{proposition}[Dominant-strategy defection under bounds]\label{prop:ds}
With $C_j\in[0,M]$ and the linear rule ($\alpha=1$), for any fixed $C_{-j}$ the map $C_j\mapsto \pi_j(C_j,C_{-j})$ is nondecreasing; 
thus $C_j^\star=M$ is a best reply independent of $C_{-j}$. (See \emph{Appendix C}.)
\end{proposition}

\begin{theorem}[Uniqueness of boundary continuity]\label{thm:uniq}
Within the $\alpha$-power family for the slack allocation mechanism, continuity at $X=I$ for all positive overage vectors holds iff $\alpha=1$.
(\emph{Appendix A}.)
\end{theorem}

\begin{theorem}[DS is Strong Nash under TU, therefore CPNE]\label{thm:SNE}
At the dominant-strategy outcome (all defect), under transferable utility within coalitions, no coalition $K$ can achieve a strict Pareto improvement by deviating; hence the profile is a Strong Nash equilibrium. By \citet{BernheimPelegWhinston1987}, Strong Nash implies Coalition-Proof Nash Equilibrium.
\end{theorem}

\begin{proof}
Let $C^{\mathrm{DS}}=(M,\dots,M)$ with $C_j\in[0,M]$ (Prop.~\ref{prop:ds}).
Then $I=0$ and by Section~2 we have $\pi_j(C^{\mathrm{DS}})=L_j$ for all $j$, hence for any coalition $K\subseteq N$,
\[
\sum_{i\in K}\pi_i(C^{\mathrm{DS}})=\sum_{i\in K}L_i.
\]

Fix any coalition $K$ and any deviation $C'_K$.
The post-deviation profile is $C'=(C'_K,\,C^{\mathrm{DS}}_{-K})$.
Since the complement $-K$ continues to defect (plays $C^{\mathrm{DS}}_{-K}$), it generates no slack; therefore at $C'$ we have $I=I_K$.

By Appendix~B (the Case~2 argument applied with a defecting complement), whenever $I=I_K$ we have the coalition payoff bound
\[
\sum_{i\in K}\pi_i(C') \;\le\; \sum_{i\in K}L_i \;=\; \sum_{i\in K}\pi_i(C^{\mathrm{DS}}).
\]
Under transferable utility, a coalition deviation can make all its members weakly better and at least one strictly better only if its total payoff strictly increases.
The bound shows this is impossible from $C^{\mathrm{DS}}$. Hence $C^{\mathrm{DS}}$ is a Strong Nash equilibrium under TU.
Since every Strong Nash equilibrium is coalition-proof \citep{BernheimPelegWhinston1987}, $C^{\mathrm{DS}}$ is also a CPNE.
\end{proof}

\begin{remark}[Equilibrium multiplicity under weak dominance]
The maximal-claiming profile $C^{\mathrm{DS}}=(M,\dots,M)$ is a dominant-strategy equilibrium. Because dominance is weak, other Nash equilibria exist. In particular, the cooperative profile $C=L$ is a Nash equilibrium: for any $j$ and any $C_j'\ge L_j$, the induced $I'=0$ yields $\pi_j(C_j',C_{-j})=L_j=\pi_j(L_j,C_{-j})$, so unilateral deviations are not profitable. Our welfare and coalition-proofness results are stated for the dominant-strategy outcome.
\end{remark}

\begin{remark}[Transferable utility (TU)]
We use TU in the standard sense: coalition members can make budget-balanced side-payments among themselves, so a deviation is evaluated by the coalition’s total payoff. Formally, for $K\subseteq N$ a deviation from $C$ to $C'$ is feasible under TU iff there exist transfers $(t_i)_{i\in K}$ with $\sum_{i\in K} t_i=0$ such that $\pi_i(C')+t_i\ge \pi_i(C)$ for all $i\in K$, with strict inequality for at least one member.
\end{remark}

\section{Policy Implementation: Cap-and-Share with overage clearing}

Interpret $L_j$ as allowances under cap-and-share (we index by $t$ only in this section and Appendix~D), and take the linear rule ($\alpha=1$).
In applications, the bound $M$ can represent a physical capacity, a regulatory limit, or a credit constraint; none of the results use more than $M>\max_j L_j$ (Appendix~C).

A forward market clears expected buy/sell orders; at period end, realized emissions induce $(v_j,s_j)$ and clearing \eqref{eq:clear}.
Residual overage $(v_j-\widehat v_j)_+$ is priced by a penalty collar $\kappa_t\in[\underline\kappa,\overline\kappa]$ modulated by an endogenous scarcity factor $\Lambda_t$ (defined below).
Appendix~D (Prop.~\ref{prop:D-expected-mc}) shows that the mechanism is compatible with forward trading without perverse incentives.

\paragraph{Penalty and scarcity.}
For period $t$, define the scarcity factor by
\[
\Lambda_t \;:=\;
\begin{cases}
0,& X^t=0,\\[3pt]
\max\{0,(X^t-I^t)/X^t\},& X^t>0,
\end{cases}
\qquad \Lambda_t\in[0,1].
\]
The per-unit penalty $\kappa_t\in[\underline\kappa,\overline\kappa]$ is regulator-set (exogenous), while $\Lambda_t$ is endogenous (determined by realized $(X^t,I^t)$).
If $X^t\le I^t$, residual overage is fully covered (zero penalty); if $X^t>I^t$, a marginal unit of overage faces at least $\kappa_t\,\Lambda_t$ in penalty at clearing (Appendix D).

\paragraph{Risk \& Governance.}
The authority should emphasize robustness and auditability rather than discretion: (i) stress-test reporting and clearing against strategic misreporting and timing manipulation; (ii) publish, ex ante, the collar calibration and adjustment protocol (data sources and decision rules); (iii) monitor realized $(X^t,I^t,\Lambda_t)$ against stated tolerances with review triggers for threshold breaches; and (iv) reserve a narrowly circumscribed emergency suspension rule that preserves budget balance and does not create profitable anticipatory deviations.

\appendix

\section*{Appendix A. Overage-power family and boundary continuity}

We generalize the slack-allocation mechanism by introducing an exponent $\alpha>0$ on overage shares.
Players choose claims $C_j\ge 0$ against entitlements $L_j>0$ and we set, as in the main text,
\[
v_j=(C_j-L_j)_+,\qquad s_j=(L_j-C_j)_+,\qquad X=\sum_m v_m,\qquad I=\sum_m s_m.
\]
Given a profile $C$, the clearing rule with exponent $\alpha$ allocates cooperative surplus $I$ to defectors ($v_j>0$) via

\begin{equation}\label{eq:A-clearing}
\widehat v_j^\alpha(C)=
\begin{cases}
v_j, & X\le I,\\[6pt]
\dfrac{I\,v_j^\alpha}{\sum_{m:v_m>0} v_m^\alpha}, & X>I,
\end{cases}
\end{equation}

and cooperators ($C_j\le L_j$) receive their claim while defectors receive entitlement plus covered overage:
\[
\pi_j^\alpha(C)=
\begin{cases}
C_j, & C_j\le L_j,\\[3pt]
L_j+\widehat v_j^\alpha(C), & C_j>L_j.
\end{cases}
\]

\paragraph{Boundary convention at $X=I$.}
At the boundary we adopt the $X<I$ branch, i.e., $\widehat v_j^\alpha=v_j$.
This preserves feasibility $\sum_{j:C_j>L_j}\widehat v_j^\alpha=I$ and budget balance.
For $\alpha\neq 1$ the allocation generally differs from the $X>I$ branch, so a discontinuity remains; for $\alpha=1$ the two branches coincide and the rule is continuous.

\paragraph{Aggregate identity.}
Let $\Pi^\alpha(C):=\sum_j \pi_j^\alpha(C)$. Since $\sum_{j:C_j\le L_j} C_j=\sum_j \min\{C_j,L_j\}$ and $\sum_{j:C_j>L_j}\widehat v_j^\alpha=\min\{X,I\}$, we have
\begin{equation}\label{eq:A-budget}
\Pi^\alpha(C)=\sum_j L_j - \max\{I-X,0\}.
\end{equation}
Hence in the scarcity region $X\ge I$ we have $\Pi^\alpha(C)=\sum_j L_j$, and in the slack region $X<I$ we have $\Pi^\alpha(C)=\sum_j L_j-(I-X)$. This coincides with the main-text budget identity \eqref{eq:budget}.

\paragraph{Monotonicity in own claim under bounds.}
Fix $C_{-j}$ and restrict $C_j\in[0,M]$.
Set $y:=(C_j-L_j)_+\in[0,M-L_j]$, $X=X_{-j}+y$, and $I=I_{-j}$; define $S_{-j}:=\sum_{m\ne j} v_m^\alpha$ (note: $S_{-j}\ne X_{-j}^\alpha$ in general).
Then
\[
\pi_j^\alpha(y)=
\begin{cases}
L_j+y, & X_{-j}+y\le I_{-j},\\[8pt]
L_j+\dfrac{I_{-j}\, y^\alpha}{S_{-j}+ y^\alpha}, & X_{-j}+y> I_{-j},
\end{cases}
\]

with the boundary value at $X_{-j}+y=I_{-j}$ set by the $X<I$ branch (boundary convention above).

On $[0,(I_{-j}-X_{-j})_+]$ we have $\pi_j^\alpha(y)=L_j+y$ which is strictly increasing.
On $((I_{-j}-X_{-j})_+,\; M-L_j]$ we have
\[
\phi(y):=L_j+\frac{I_{-j}\, y^\alpha}{S_{-j}+ y^\alpha},
\qquad
\phi'(y)=\frac{\alpha\, I_{-j}\, S_{-j}\, y^{\alpha-1}}{(S_{-j}+ y^\alpha)^2}\;\ge 0,
\]
so $\pi_j^\alpha$ is nondecreasing there. For $\alpha=1$, the branches coincide at the boundary, hence $\pi_j^1(\cdot,C_{-j})$ is nondecreasing on $[0,M]$.
In particular a best reply is attained at $C_j^\star=M$ for $\alpha=1$.

\begin{theorem}[Uniqueness of boundary continuity]\label{thm:A-uniq}
For \eqref{eq:A-clearing}, the payoff map $C\mapsto \pi^\alpha(C)$ is continuous at profiles with $X=I$ for all positive overage vectors if and only if $\alpha=1$.
\end{theorem}

\begin{proof}
Continuity away from $X=I$ is immediate.
At $X=I$ with positive $v$, approaching from $X<I$ gives $\widehat v_j^\alpha=v_j$ while from $X>I$ gives $\widehat v_j^\alpha=I\,v_j^\alpha / \sum_\ell v_\ell^\alpha$.
Equality for all positive $v$ forces
\(
\frac{v_j^\alpha}{\sum_\ell v_\ell^\alpha}=\frac{v_j}{\sum_\ell v_\ell}
\),
which holds iff $\alpha=1$; conversely, for $\alpha=1$ the branches coincide.
The boundary convention does not affect this characterization.
\end{proof}

\begin{remark}[Uniqueness of the linear rule at the boundary]
On the scarcity side $X>I$, consider the $\alpha$-weighted alternative that rations overages by weights proportional to $v_j^{\alpha}$:
\[
\pi_j^{(\alpha)} \;=\; L_j \;+\; I\,\frac{v_j^{\alpha}}{\sum_i v_i^{\alpha}}.
\]
On the slack side $X<I$, payoffs are $\pi_j=L_j+v_j$. At the boundary $X=I$ the one-sided limits match for all overage vectors $v$ if and only if $\alpha=1$. For any $\alpha\neq1$, equality fails except on the measure-zero set where all positive $v_j$ are equal, so a jump occurs at $X=I$.
\end{remark}

\section*{Appendix B. Coalition-proofness at the DS outcome}

We prove coalition-proofness for the linear rule ($\alpha=1$) under slack allocation.

\begin{proof}
Let $K$ be any coalition and assume the complement $-K$ \emph{defects}.
We show that the coalition's aggregate payoff at any deviation $C_K$ cannot exceed $\sum_{i\in K} L_i$.

Let $S:=\{\,i: C_i\le L_i\,\}$ (cooperators) and $D:=\{\,i: C_i> L_i\,\}$ (defectors). Write $X_K=\sum_{j\in K} v_j$, $X_{-K}=\sum_{j\notin K} v_j$, $X=X_K+X_{-K}$, and $I_K=\sum_{i\in K} s_i$. Since $-K$ defects, their slack is $0$, so $I=I_K$.

The coalition's aggregate payoff is derived as follows. Note that for $i\in S$, $C_i = L_i - s_i$.
\begin{align*}
\sum_{i\in K}\pi_i
&=\sum_{i\in K\cap S} C_i \;+\; \sum_{j\in K\cap D} \big( L_j + \widehat v_j\big)\\
&=\sum_{i\in K\cap S} (L_i - s_i) \;+\; \sum_{j\in K\cap D} L_j \;+\; \sum_{j\in K\cap D} \widehat v_j \\
&=\sum_{i\in K} L_i \;-\; \sum_{i\in K\cap S} s_i \;+\; \sum_{j\in K\cap D} \widehat v_j.
\end{align*}
Since $s_i=0$ for $i\in D$, we have $\sum_{i\in K\cap S} s_i = \sum_{i\in K} s_i = I_K$. Thus,
\begin{equation}\label{eq:B-coalition-payoff}
\sum_{i\in K}\pi_i = \sum_{i\in K} L_i \;-\; I_K \;+\; \sum_{j\in K\cap D} \widehat v_j.
\end{equation}
We bound the last term in the two regions.

\emph{Case 1: $X\le I$.}
All overage is covered, so $\sum_{j\in K\cap D}\widehat v_j = X_K$.
Then by \eqref{eq:B-coalition-payoff},
\(
\sum_{i\in K}\pi_i = \sum_{i\in K} L_i + (X_K - I_K).
\)
But $X\le I$ and $I=I_K$ imply $X_K \le I_K - X_{-K}\le I_K$, hence $X_K - I_K \le -X_{-K}\le 0$. Therefore $\sum_{i\in K}\pi_i \le \sum_{i\in K} L_i$.

\emph{Case 2: $X> I$.}
Coverage is proportional: $\sum_{j\in K\cap D}\widehat v_j = I \cdot \frac{X_K}{X}= I_K \cdot \frac{X_K}{X}$.
Thus by \eqref{eq:B-coalition-payoff},
\begin{align*}
\sum_{i\in K}\pi_i &= \sum_{i\in K} L_i - I_K + I_K \tfrac{X_K}{X} \\
&= \sum_{i\in K} L_i - I_K \big(1-\tfrac{X_K}{X}\big) \\
&= \sum_{i\in K} L_i - I_K \tfrac{X_{-K}}{X} \;\le\; \sum_{i\in K} L_i. \qedhere
\end{align*}
\end{proof}

\section*{Appendix C. Bounded-action regularization}

Fix $M>0$. Each player $j$ chooses $C_j\in[0,M]$.

\emph{Assumption (large action bound).} Throughout Appendix~C and any results that invoke it, take $M>\max_j L_j$, so that the maximal claim $M$ constitutes defection; if this fails, replace “defection” with “maximal claim” in the statements without altering the analysis.

Let $v_j=(C_j-L_j)_+$, $s_j=(L_j-C_j)_+$, $X=\sum_m v_m$, $I=\sum_m s_m$.
For $\alpha=1$, $\widehat v_j=v_j$ if $X\le I$ and $\widehat v_j=(I/X)v_j$ if $X>I$.
Payoffs are $\pi_j=C_j$ when $C_j\le L_j$, and $\pi_j=L_j+\widehat v_j$ when $C_j>L_j$.
\medskip

\noindent\textbf{Proposition C.} (i) Best replies exist. (ii) For any fixed $C_{-j}$, $\pi_j(C_j,C_{-j})$ is nondecreasing on $[0,M]$, hence a best reply is $C_j^\star=M$.

\begin{proof}
Fix $C_{-j}$. We first establish continuity. On $[0,L_j]$, $\pi_j(C_j,C_{-j})=C_j$.
For the region $C_j\ge L_j$, write $y:=(C_j-L_j)_+\in[0,M-L_j]$, $X=X_{-j}+y$, $I=I_{-j}$. Then
\[
\pi_j(y)=
\begin{cases}
L_j+y, & X_{-j}+y\le I_{-j},\\[4pt]
L_j+\dfrac{I_{-j}y}{X_{-j}+y}, & X_{-j}+y>I_{-j}.
\end{cases}
\]
We check continuity at the switching point $y^\star=(I_{-j}-X_{-j})_+$.
If $y^\star>0$, then $y^\star = I_{-j}-X_{-j}$, which implies $X_{-j}+y^\star = I_{-j}$. The two branches evaluate to:
\begin{align*}
\text{Branch 1 (slack):} &\quad L_j+y^\star, \\
\text{Branch 2 (scarcity):} &\quad L_j + \frac{I_{-j}y^\star}{X_{-j}+y^\star} = L_j + \frac{I_{-j}y^\star}{I_{-j}} = L_j+y^\star.
\end{align*}
If $y^\star=0$ (i.e., $I_{-j}\le X_{-j}$), the payoff function near $y=0$ is continuous, and both expressions evaluate to $L_j$ at $y=0$. Thus the branches agree at $y^\star$, and $\pi_j(\cdot,C_{-j})$ is continuous on $[0,M]$.

(i) Since $\pi_j(\cdot,C_{-j})$ is continuous on the compact interval $[0,M]$, a maximizer exists by the Weierstrass extreme value theorem.

(ii) We establish monotonicity. On $[0,L_j]$, $\pi_j$ is strictly increasing. For $C_j > L_j$: On $[0,y^\star]$, $\pi_j(y)=L_j+y$ is strictly increasing. On $(y^\star,M-L_j]$,
$\frac{\mathrm d}{\mathrm dy}\!\left(L_j+\frac{I_{-j}y}{X_{-j}+y}\right)=\frac{I_{-j}X_{-j}}{(X_{-j}+y)^2}\ge0$,
so $\pi_j$ is nondecreasing. Hence $\pi_j$ is nondecreasing on $[0,M]$ and a best reply is $C_j^\star=M$. \qedhere
\end{proof}

\section*{Appendix D. Penalty-Collar Economics: No Gain from Strategic Over-Emission}

By a penalty collar we mean a regulated interval $[\underline\kappa,\overline\kappa]$ for the per-unit penalty applied to uncovered residual overage at clearing; the realized period-$t$ penalty is $\kappa_t\in[\underline\kappa,\overline\kappa]$.

\paragraph{Setup and notation.}

Fix period $t$. Each entity $i$ has entitlement $L_i^t>0$ and realizes usage (claims) $C_i^t\ge0$.
Define overage $v_i:=(C_i^t-L_i^t)_+$, slack $s_i:=(L_i^t-C_i^t)_+$, aggregate $X^t=\sum_i v_i$, $I^t=\sum_i s_i$.
For any $j$, let $V_{-j}:=\sum_{m\ne j} v_m$ denote the aggregate overage of others, so $X^t=v_j+V_{-j}$.

The scarcity factor $\Lambda_t$ is defined in Section 4.
End-of-period clearing covers defectors' overage proportionally: $\widehat v_i = v_i$ if $X^t\le I^t$, else $\widehat v_i = (I^t/X^t)\,v_i$.
Residual overage is $r_i:= (v_i-\widehat v_i)_+$. Note that $r_i = \Lambda_t v_i$ when $X^t>I^t$ and $0$ otherwise.

Let $p_\tau$ be the forward price at decision time $\tau<t$, conditional on information $\mathcal F_\tau$.

\paragraph{Prices and calibration parameters.}
All expectations below are conditional on the information set defined in the next paragraph.
Let $p_t$ denote the period-$t$ spot price at clearing, and let $\overline p_t$ be a publicly announced upper bound on $p_t$ (e.g., an auction reserve or penalty ceiling). \emph{Assumption (expected scarcity).} There exists $\underline\lambda\in(0,1]$ such that $\mathbb E[\Lambda_t\mid \mathcal F_\tau]\ge \underline\lambda$. This assumption is used only in the collar-calibration corollary below.

\paragraph{Information set.}
For period $t$, let $\mathcal F_\tau$ denote the public information available by decision time $\tau<t$: (i) entitlements $\{L_i^t\}$; (ii) policy parameters $(\underline\kappa,\overline\kappa)$; (iii) forward orders/positions and any other public signals observed by $\tau$ that bear on the period-$t$ aggregates $(X^t,I^t)$ and on the realized penalty $\kappa_t$. Expectations $\mathbb E[\cdot\mid \mathcal F_\tau]$ are conditional on that information. At clearing, $\Lambda_t$ and $\kappa_t$ are $\mathcal F_t$-measurable (but need not be $\mathcal F_\tau$-measurable for $\tau<t$).

\begin{proposition}[Expected marginal cost of waiting]\label{prop:D-expected-mc}
For any defector $j$ with overage $v_j$ at time $t$, the expected unit cost of creating one additional unit by waiting for clearing is at least
\[
\mathbb E\!\left[\kappa_t\,\Lambda_t \,\middle|\, \mathcal F_\tau\right].
\]
\end{proposition}

\begin{proof}
Write $P_j(v_j):=\kappa_t\,r_j(v_j)$ for $j$'s penalty at clearing, where $r_j=(v_j-\widehat v_j)_+$. The expected unit cost of creating one more unit by waiting equals the conditional expectation of the right marginal $\partial P_j/\partial v_j$ holding $(V_{-j},I^t)$ fixed.

If $X^t\le I^t$, then $\Lambda_t=0$ in a neighborhood and $r_j\equiv 0$, so $\frac{\partial r_j}{\partial v_j}=0=\Lambda_t$, hence
$\frac{\partial P_j}{\partial v_j}=\kappa_t\,\frac{\partial r_j}{\partial v_j}\ge \kappa_t\,\Lambda_t$.

If $X^t>I^t$, then $\Lambda_t=\frac{(X^t-I^t)}{X^t}=1-\frac{I^t}{X^t}$ with $X^t=v_j+V_{-j}$. Differentiating w.r.t.\ $v_j$ gives
\[
\frac{\partial \Lambda_t}{\partial v_j}
=\frac{I^t}{(X^t)^2}\;\ge\;0.
\]
Since $r_j=\Lambda_t\,v_j$ in this region,
\[
\frac{\partial r_j}{\partial v_j}
=\Lambda_t + v_j\,\frac{\partial\Lambda_t}{\partial v_j}
\;\ge\;\Lambda_t,
\]
and therefore $\frac{\partial P_j}{\partial v_j}
=\kappa_t\,\frac{\partial r_j}{\partial v_j}
\ge \kappa_t\,\Lambda_t$.

At the kink $X^t=I^t$ the right derivative exists and the same inequality holds by the above cases. Taking conditional expectations yields
\[
\mathbb{E}\!\left[\frac{\partial P_j}{\partial v_j}\,\middle|\,\mathcal F_\tau\right]
\;\ge\;
\mathbb{E}\!\left[\kappa_t\,\Lambda_t\,\middle|\,\mathcal F_\tau\right],
\]
which proves the claim.
\end{proof}

\begin{lemma}[No benefit from inflating overage]\label{lem:D-no-gain}
Fix $v_{-j}$ and $I^t$ with $X^t>I^t$.
With $X^t=v_j+V_{-j}$ and $\widehat v_j=I^t\,v_j/X^t$,
\[
\frac{\mathrm d}{\mathrm d v_j} r_j(v_j)
= 1 - \frac{\mathrm d}{\mathrm d v_j}\Big(\frac{I^t v_j}{X^t}\Big)
= 1 - I^t\frac{X^t - v_j}{(X^t)^2}
\ge 1-\frac{I^t}{X^t}
=\Lambda_t,
\]
holding $I^t, V_{-j}$ fixed. Thus $r_j$ is strictly increasing and the incremental penalty is at least $\kappa_t\,\Lambda_t$ per unit.
\end{lemma}

\begin{corollary}[Collar calibration kills ``wait-and-emit'' arbitrage]\label{cor:D-collar}
If the authority sets $\kappa_t\ge \overline p_t$ (auction reserve or price cap) and publishes $\mathbb E[\Lambda_t\mid\mathcal F_\tau]\ge \underline\lambda>0$, then
$\mathbb E[\kappa_t\,\Lambda_t\mid \mathcal F_\tau]\ge \overline p_t\,\underline\lambda$.
If $\overline p_t\,\underline\lambda \ge p_\tau$, forward purchase is weakly cheaper in expectation than waiting; strict if $>$.
\end{corollary}

\paragraph{Budget balance reminder.}
When $X^t\ge I^t$, the clearing is budget balanced, i.e., $\sum_i \pi_i^t=\sum_i L_i^t$ (cf. Eq. \eqref{eq:budget}; here $\pi_i^t$ denotes the period-$t$ payoff.

\bibliographystyle{unsrtnat}
\bibliography{efficient_defection_v22} 

\end{document}